\RequirePackage[]{silence}
\documentclass[11pt]{article}  

\usepackage[latin1]{inputenc}
\usepackage[dvips]{graphicx}
\usepackage{color}
\usepackage{graphics} 

\usepackage[letterpaper,margin=1.15in]{geometry}
\usepackage[pagebackref=false]{hyperref}
\hypersetup{
    unicode=false,          
    colorlinks=true,        
    linkcolor=red,          
    citecolor=green,        
    filecolor=magenta,      
    urlcolor=cyan           
}

\usepackage{cleveref}

\usepackage{amsmath}
\usepackage{amsthm}

\newtheorem{theorem}{Theorem}[section]
\newtheorem{lemma}[theorem]{Lemma}
\newtheorem{corollary}[theorem]{Corollary}
\newtheorem{definition}{Definition}[section]

\newlength{\figtxtwid}
\long\gdef\boxitnew#1{\dimen200 = \hsize \advance\dimen200 by -7pt
\begingroup\vbox{\hrule \hbox to \hsize{\vrule\kern3pt
      \vbox{\hsize \dimen200 \kern3pt#1\kern3pt}\hfil \kern3pt\vrule}\hrule}\endgroup}

\newlength{\myfigwidth}
\newcommand{\boxtext}[1]{\setlength{\myfigwidth}{\columnwidth}\addtolength{\myfigwidth}{-1ex}\boxitnew{\vspace*{.1ex}\begin{minipage}{\myfigwidth}\setlength{\parindent}{0pt}#1\end{minipage}\\[.1ex]}}

\renewcommand{\boxtext}[1]{\setlength{\parindent}{0pt}#1\\[.1ex]}

\renewcommand{\boxtext}[1]{\begin{minipage}{\columnwidth}\setlength{\parindent}{0pt}#1\end{minipage}}

\renewcommand{\boxtext}[1]{\boxitnew{\setlength{\parindent}{0pt}\vspace{-0.2cm} #1 \vspace{-0.2cm}}}

\newcommand{\polylog}{\textrm{ polylog }}

\newcommand{\ceil}[1]{\left\lceil #1 \right\rceil}

\newcommand{\IncludePictures}[1]{#1}
\renewcommand{\IncludePictures}[1]{}

\newcommand{\fullOnly}[1]{}

\begin{document}

\title{Simple, Fast and Deterministic Gossip and Rumor Spreading\footnote{The conference version of this paper was published in the ACM-SIAM Symposium for Discrete Algorithms~\cite{haeupler2013simple} where it won a best student paper award.}}
\author{Bernhard Haeupler\\ Microsoft Research \\ \texttt{haeupler@cs.cmu.edu}}
\date{}

\maketitle

\begin{abstract}
We study gossip algorithms for the rumor spreading problem, which asks each node to deliver a rumor to all nodes in an unknown network. Gossip algorithms allow nodes only to call one neighbor per round and have recently attracted attention as message efficient, simple, and robust solutions to the rumor spreading problem.

\smallskip

A long series of papers analyzed the performance of uniform random gossip in which nodes repeatedly call a random neighbor to exchange all rumors with. A main result of this line of work was that uniform gossip completes in $O(\frac{\log n}{\Phi})$ rounds where $\Phi$ is the conductance of the network. More recently, non-uniform random gossip schemes were devised to allow efficient rumor spreading in networks with bottlenecks. 
In particular, [Censor-Hillel et al., STOC'12] gave an $O(\log^3 n)$ algorithm to solve the $1$-local broadcast problem in which each node wants to exchange rumors locally with its $1$-neighborhood. By repeatedly applying this protocol, one can solve the global rumor spreading quickly for all networks with small diameter, independently of the conductance.

\smallskip

All these algorithms are inherently randomized in their design and analysis. A parallel research direction has been to reduce and determine the amount of randomness needed for efficient rumor spreading. This has been done via lower bounds for restricted models and by designing gossip algorithms with a reduced need for randomness, e.g., by using pseudorandom generators with short random seeds. The general intuition and consensus of these results has been that randomization plays a important role in effectively spreading rumors and that at least a polylogarithmic number of random bit are crucially needed.

\smallskip

In this paper we improves over this state of the art in several ways by presenting a \emph{deterministic} gossip algorithm that solves the the {$k$-local} broadcast problem in $2(k+\log n) \log n$ rounds\footnote{Throughout this paper $\log x$ denotes $\ceil{\log_2 x}$, that is, the rounded up binary logarithm.}. Besides being the first efficient deterministic solution to the rumor spreading problem this algorithm is interesting in many aspects: It is simpler, more natural, more robust, and faster than its randomized pendant and guarantees success with certainty instead of with high probability. Its analysis is furthermore simple, self-contained, and fundamentally different from prior works.

\end{abstract}

\newpage

\section{Introduction}

Broadcasting, that is, disseminating information present initially at different nodes in an unknown network to every node, is a fundamental network communication primitive with many applications.
It has been studied under different names such as gossip, rumor spreading, information dissemination, (all-to-all) multicast, and (global) broadcast.
\smallskip

\emph{Gossip algorithms}, during which nodes contact or call only one neighbor at a time, have been proposed as a powerful time and message efficient alternatives to flooding, that is, repeatedly forwarding information to all neighbors, or structured broadcast protocols, which often require a stable network with known topology. 

\smallskip

The simplest and most widely studied form of gossip is \emph{uniform (random) gossip} in which nodes repeatedly call a random neighbor to exchange information. A series of results showed that this algorithm performs well on well-connected graphs with no bottleneck(s)~\cite{frieze1985shortest,pittel1987spreading,uniform1,uniform2,uniformfinal,soda2012_vertex-expansion,vertexexpansion2}. More precisely, the main result is a tight bound of $O(\frac{\log n}{\Phi})$ rounds, where $\Phi$ is the conductance of the network. More recently, non-uniform random gossip schemes were devised to allow efficient rumor spreading in arbitrary networks~\cite{weakcond,random}. The \emph{local broadcast} problem, that asks each node to exchange rumors locally with all its neighbors, has been a crucial abstraction to obtain results independent from any conductance-type measure. In particular, building on the results on uniform gossip it was shown in \cite{random} how to solve the local broadcast problem in $O(\log^3 n)$ rounds. Repeatedly applying this solution leads to an $O(D \log^3 n)$ global broadcast in any network with diameter $D$ and thus to a polylogarithmic time gossip solution for any network with polylogarithmic diameter. Using connections to spanners one can furthermore get a $O(D + \log^{O(1)})$ solution.
\smallskip

All these algorithms are inherently randomized in both their design and analysis in that they crucially rely on the effect that choosing neighbors randomly for forwarding disperses information quickly in expanders or expanding subgraphs. A parallel research direction to finding faster and more general gossip algorithms has been to study the necessity for this randomization. In particular, there are both lower bounds quantifying how much randomness is inherently needed for efficient gossip algorithms \cite{quasirandom-randomnesstradeoff,lowerboundsrand,KarpSSV2000} and newly designed gossip protocols that work with reduced amounts of randomness \cite{DoerrFS08,DoerrFS09,FountoulakisH2009,lowrandomness}. The general intuition and also the consensus of these works has been that (some) randomization plays a crucial role in efficiently spreading rumors.

\subsection{Our Results}

This paper contributes to both of these research directions by presenting a fast, simple, natural, robust and deterministic gossip algorithm for the local broadcast problem:

\begin{theorem}\label{thm:main}
For any $k$ there is a simple deterministic gossip algorithm that runs for $2(k \log n + \log^2 n)$ rounds on any $n$-node network and solves the {$k$-local} broadcast problem, that is, allows each node to exchange a rumor with each node at distance at most $k$.
\end{theorem} 

The next corollary shows that this directly implies a fast and simple algorithm for the global broadcast problem as well. We remark that using the connection to spanners from \cite{random} one can also obtain a (significantly less simple) deterministic gossip algorithm solving the global broadcast problem in $O(D + \log^{O(1)} n)$ rounds, given that there exists a deterministic variant of the spanner construction from \cite{Pettie2009}. 

\begin{corollary}\label{cor:main1}
There is a simple deterministic gossip algorithm that runs for $2(D \log n + \log^2 n)$ rounds on any $n$-node network with diameter $D$ and solves the global broadcast problem.
\end{corollary}

Certainly, the most striking aspect of \Cref{thm:main} is that it constitutes the first deterministic gossip algorithm for the broadcast setting studied here and in~\cite{weakcond,random}. We feel that this is an  interesting and surprising result given the general intuition and the results from~\cite{quasirandom-randomnesstradeoff,lowerboundsrand,KarpSSV2000,lowrandomness} that at least some randomness is needed to enable an efficient gossip algorithm.  We emphasize that our result does not stem from unusual model assumptions. Exactly as in, e.g., \cite{weakcond,random} we only assume that (1) each node only knows the IDs of its neighbors and (2) two nodes involved in a call can exchange all rumors known to them. 

Our algorithm has several other advantages. For one it is faster than previous randomized algorithms: For $k=1$ our algorithm is logarithmically faster than \cite{random} while for $k=\log n$ its $4\log^2 n$ running time is a $\Theta(\log^2 n)$-factor and therefore quadratically faster. Even when compared to the performance of uniform gossip it is at most a $\log n$-factor slower on expanders while continuing to be near optimal on any other topology. Furthermore, any faster (randomized) algorithm with running time $O(\log^{2-\epsilon} n)$ for $k=\log n$ would imply the existence of spanners with better than known quality (see Conclusion).	Lastly, its deterministic nature brings with it the advantage that these running time guarantees and the algorithm's correctness hold with certainty instead of with high probability.
	
Our algorithm is also simpler, more natural and more robust. In fact, one of our algorithms is a simplification of \cite{random} in which we allow all randomized choices to be replaced by arbitrary ones. While our analysis is powerful enough to work in this more general setting it remains simple, short, and self-contained. It also extends nicely to analyzing a wide variety of related natural processes demonstrating the robustness of both the algorithm and its analysis.

\subsection{Related Work}

Gossip and rumor spreading have been intensely studied both for the setting of a single rumor being spread and for a rumor being spread from each node in the network. A difference between the two settings becomes mostly apparent when the amount of information exchanged between nodes in a round is limited. In such a scenario very different techniques like algebraic gossip~\cite{HaeuplerNC} become interesting. In this paper we assume large packet sizes that allow two nodes to exchange potentially all rumors in one packet. In this setting it typically does not matter whether one or more rumors are to be spread. 
\smallskip 


The spreading of a rumor according to the uniform gossip process was first considered by Frieze and Gimmet~\cite{frieze1985shortest} and subsequently Pittel~\cite{pittel1987spreading} which proved that only $(1 + \ln 2) \log n + O(1)$ rounds are needed on the complete graph. Lower bounds and non-uniform algorithms for the complete graph were investigated by~\cite{KarpSSV2000}. Going to more general topologies~\cite{uniform1} showed that uniform gossip works well in any expander. More precisely, an $O(\frac{\log^4 n}{\Phi^6})$ spreading time was proven for any graph with conductance $\Phi$. This bound was then improved to $O(\frac{\log^2 \Phi^{-1} \log n}{\Phi})$~\cite{uniform2} and finally $O(\frac{\log n}{\Phi})$~\cite{uniformfinal}. This bound is tight since for any $\Phi = \Omega(1/n)$ there are graphs with conductance $\Phi$ but diameter $\Theta(\frac{\log n}{\Phi})$ on which uniform gossip requires $\Theta(\frac{\log n}{\Phi})$ rounds to complete. Similar results were obtained for the spreading time of uniform gossip on networks with vertex expansion $\alpha$. In particular, \cite{soda2012_vertex-expansion,vertexexpansion2} proved an $O(\frac{\log^2 n}{\alpha})$ bound and showed it to be tight as well. These results demonstrate very nicely that uniform gossip performs at least as good as the worst bottleneck that can be found in the network. On the other hand it is also relatively easy to come up with networks that have a small diameter but a bad bottleneck, in form of a bad cut, on which uniform gossip fails badly due to not being able to bridge this bottleneck. All in all, this gives a very tight understanding of what can and cannot be achieved with uniform gossip.

\smallskip 

Several papers have given improvements over the uniform gossip protocol that cope better with bottlenecks. In~\cite{weakcond} a rumor spreading algorithm was provided which works well for any graph with good weak expansion. This includes some graphs with (few) bottlenecks. The first algorithm to work efficiently on any topology was~\cite{random}. As already summarized in the introduction this paper gives a $O(\log^3 n)$ algorithm for the local broadcast problem as a crucial step to achieve a polylogarithmic running time on any network with polylogarithmic diameter. For all algorithms up to then there existed networks with very small diameter, e.g., $O(\log n)$, on which gossip took $\Omega(n)$ time.  

\smallskip 

The other research direction that is relevant for this paper aimed at determining and reducing the amount of randomness required by efficient gossip algorithms. A very interesting and successful way to reduce the amount of randomness is the quasirandom rumor spreading process of~\cite{DoerrFS08}. This protocol assumes an arbitrary cyclic ordering of neighbors at each nodes but randomizes it by picking a random random starting point. It was shown that this simple algorithm achieves a similar performance as the fully random uniform gossip for many topologies like the complete graph~\cite{FountoulakisH2009,quasirandom-tight}, random graphs, 
hypercubes or expanders~\cite{DoerrFS09}. The quasirandom rumor spreading process only requires $O(\log n)$ bits of randomness per node to pick the starting position. In~\cite{quasirandom-randomnesstradeoff} it was shown that one cannot further reduce this amount without a severe loss of efficiency. In particular, if one uses $o(\log n)$ random bits to choose uniformly at random between a subset of equidistant starting points then the number of rounds becomes almost linear instead of logarithmic~\cite{quasirandom-randomnesstradeoff}. The question how much randomness suffices for gossip to be efficient was also addressed in~\cite{lowerboundsrand}. This paper presents an algorithm that uses a total of $n \log \log n$ bits of randomness, gives a non-constructive argument for the existence of a gossip algorithm with roughly $2\log n$ bits of randomness and shows that no algorithm from a natural class of gossip algorithms can use less than logarithmic amount of randomness without taking roughly linear time. Algorithms that use pseudorandom generators or hashing schemes require only one random seed of polylogarithmic length~\cite{lowrandomness} and thus almost achieve this low total amount of randomness.

\smallskip

\paragraph{Organization}

The rest of the paper is organized as follows: In \Cref{sec:modelandproblem} we first formally define the network model and the (local) broadcast problem. 
In \Cref{sec:rand} we give a randomized local broadcast algorithm that simplifies the algorithm from \cite{random}.
In \Cref{sec:detsimple} we further simplify this algorithm and give a novel and more powerful analysis, which demonstrates
that this algorithm remains correct and efficient even when all random choices are replaced by arbitrary deterministic
choices. In \Cref{sec:more} we then show how to achieve further speed ups and obtain \Cref{thm:main}.
Lastly, in \Cref{sec:Natural Gossip Processes and Robustness} we explain how our algorithm can be interpreted as a very natural process with many
desirable robustness properties.

\section{Gossip Algorithms and the Local Broadcast Problem}\label{sec:modelandproblem}

In this section we define gossip protocols, the class of communication algorithms we are interested in, and the global and local broadcast problems.

\subsection{Gossip Algorithms}\label{sec:model}

We study gossip protocols, 
 that is, synchronous communication algorithms in which in each round each node calls (at most) one neighbor for a bidirectional information exchange. This type of algorithm fits the uniform PUSH-PULL gossip protocols that have been widely studied and matches algorithms allowed in the GOSSIP model that was defined in~\cite{random} as a restriction of the standard LOCAL model for distributed computing. The setting is given as follows:

\medskip 
{\noindent \bfseries Network:\ }

A network is specified by an undirected graph $G=(V,E)$ with node set $V$ and edge set $E$. We denote the number of nodes with $n = |V|$, the number or edges with $m = |E|$ and the diameter of $G$ with $D$. For every node $v \in V$ we define the neighborhood $\Gamma_G(v)$ to be the set of nodes whose distance to $v$ is at most one (including itself). Similarly, we define $\Gamma_G^k(v)$ to be all nodes of distance at most $k$ from $v$ and call this the $k$-neighborhood. We omit the subscript if the graph is clear from the context. 

\medskip 
{\noindent \bfseries Communication:\ }

Nodes communicate in synchronous rounds $t \in \{0,1,\ldots\}$. In each round $t$ each node chooses a message and an incident edge. We denote the union of the selected edges in a round $t$ with $E_t$ and the resulting graph with $G_t = (V,E_t)$. 
 After this selection process for any node $v$ the message of $v$ is delivered to all nodes $w \in \Gamma_{G_t}(v)$, that is, all nodes that contacted $v$ or were contacted by $v$. In short, we allow every nodes to initialize one call or bidirectional message exchange per round.

We assume that the cost of communication is completely covered by the caller, i.e., the node that initializes a contact. This leads to merely a constant cost for communication per node and round. 

Similar to the LOCAL model the GOSSIP model of~\cite{random} does in principle not limit the complexity of local computations or the size of the messages. This allows a node to always exchange all rumors it knows. Indeed, all messages sent by a node during our protocols will simply consist the collection of rumors it knows. For any node $v$ we denote this collection with $R_v$.

\medskip 
{\noindent \bfseries Initial Knowledge:\ }

We assume that each node has a unique identifier (UID). Each node initially solely knows its UID and the UIDs of its neighbors. It is important that the network topology $G$ is unknown to the nodes and we assume that no further knowledge about the network is known to the nodes.

\medskip 
{\noindent \bfseries Important Remarks:\ }\vspace{-0.1cm}
\begin{itemize}
	\item We will assume 
	that in the beginning nodes attach their UID to their rumor. In this way, each node can easily tell which of its neighbors' rumors it has already (indirectly) received. This also allows us to identify a node, its rumor and its UID and interpret the set $R_v$ of rumors known to node $v$ as a set of neighboring nodes it has (indirectly) heard from.  
	\item  While the GOSSIP model allows arbitrary local computations and message sizes our algorithms do not exploit this freedom: The local computations are extremely simple and while the set of rumors sent in a message can reach a size of $\Theta(n)$ in the worst-case, this is optimal up to a small polylogarithmic factor considering that in dense graphs the total amount of information learned by nodes during a $1$-neighborhood exchange is of order $\Omega(n^2)$ while only $O(n \polylog n)$ messages are exchanged.
	\item The UID's assumed in our model are \emph{not} used for symmetry breaking but solely to allow nodes to talk about other nodes. The only operations that are used on the UIDs is that they are routed through the network with their rumor and compared for equality to decide whether a node has already heard (indirectly) from one of its neighbors or not. 
\end{itemize}

\subsection{The Rumor Spreading and Local Broadcast Problem}\label{sec:problemdef}

The classical problem to be solved by gossip algorithms is the following global broadcast problem:

\begin{definition}[(Global) Broadcast Problem]
Each node $v$ starts with one rumor $r_v$ and the task is to inform every node about all rumors. 
\end{definition}

In this paper we mostly focus on solving the {$k$-local} broadcast problem which is an important refinement of the global broadcast problem. 

\begin{definition}[{$k$-Local} Broadcast Problem]
In the {$k$-local} broadcast problem each node $v$ starts with one rumor $r_v$ and the task is for each node $v$ to learn all rumors $r_u$ of nodes $u\in \Gamma^k(v)$ in its $k$-neighborhood. 
\end{definition}

\noindent There are several motivations to introduce and study the {$k$-local} broadcast problem:
\begin{itemize}
	\item It generalizes the global broadcast problem. In particular, the global broadcast problem is equivalent to the $n$-local broadcast problem or the {$k$-local} broadcast problem for $k \geq D$, where $D$ is the diameter of the network.
	\item As explained before, the local-broadcast problem has been proven crucial as a subproblem to solve the global broadcast problem in general topologies with bottlenecks.
	\item The local broadcast problem composes nicely. In particular, any $O(T)$ algorithm for the $k$-local broadcast problem also gives an $O(l \cdot T)$ algorithm for the $l \cdot k$-neighborhood exchange problem for any integer $l \geq 1$: Simply repeat the broadcast $l$ times. 
	\item The {$1$-local} broadcast problem and the $(\log^c n)$-local broadcast problem are natural distributed communication problems in their own right: Especially in distributed settings with large networks it is realistic that nodes are only interested in sufficiently local information.
	\item It was observed in~\cite{random} that the {$1$-local} broadcast problem corresponds to one communication step in the LOCAL model. Thus with a $T$-round gossip algorithm for the {$k$-local} broadcast problem any $T'$-round distributed algorithm in the LOCAL model can be simulated by an $\ceil{\frac{T}{k} \cdot T'}$ gossip algorithm. Many distributed problems have $O(\log^c n)$ LOCAL algorithms, that is, can be solved by each node knowing only its $O(\log^c n)$ neighborhood. This reinforces the importance of the $O(\log^c n)$-local broadcast problem. 
\end{itemize}

\medskip 
{\noindent \bfseries Remark:\ }
The GOSSIP model does not initialize nodes with any non-local knowledge about the network In particular, nodes do not know the network size $n$. For the $1$-local broadcast this is not a crucial assumption. Indeed, since it is easy to verify locally whether a 1-local broadcast completed nodes can simply guess an upper bound on $n$ and square their guess if the algorithm does not complete with their guess. For algorithms with a polylogarithmic running time squaring the guess increases the running time by a constant factor and the total running forms geometric sum which is only a constant factor larger than the final execution with the fist correct upper bound for $n$. For randomized algorithms the same verification/restart strategy can be used to transform any Monte Carlo algorithm to a Las Vegas algorithm.

\section{A Simpler Randomized {$1$-Local} Broadcast Gossip Algorithm}\label{sec:rand}

In this section we give a simple randomized gossip algorithms for the {$1$-local} broadcast problem. The algorithm can be seen as a simplification of the algorithm given in~\cite{random}. To analyze this algorithm we use the methods of~\cite{random} namely a decomposition of any graph into expanders and the efficiency of random gossip in expanders. In \Cref{sec:detsimple} we will then further strip down this algorithm and show that replacing its random choices by arbitrary deterministic choices does not affect its correctness or efficiency.


\subsection{Round Robin Flooding}

The gossip algorithm in this section uses a simple round robin flooding subroutine which we introduce here first.

Suppose all nodes have established links to at most $\Delta$ neighbors. It is quite straight forward to flood information along these links in $\Delta$ rounds by each node exchanging information over its links one by one. Essentially repeating this $d$ times floods messages for $d$-hops along all established links in $d \Delta$ steps. For completeness we add the exact statement and algorithm for this flooding procedure next:

\begin{lemma}\label{lem:flooding}
Suppose each node $v$ knows the rumors $R_v$ and has selected $\Delta_v$ links to nodes $n_v(1), \ldots, n_v(\Delta_v)$. Suppose also, the distance $d$ and an upper bound of $\Delta$ on $\max_u \Delta_u$ is given to every node.\\[0.1cm]
Then, Algorithm $1$ spreads each rumor for $d$ hops along the selected links in $\Delta d$ rounds.\\[0.1cm]
That is, each node $v$ knows exactly the rumors in $\bigcup_{u \in \Gamma_{G'}^{d}(v)} R_u$ after termination where $G' = (V, E')$ is the undirected graph with $\displaystyle E'=\bigcup_{v,i \leq \Delta_v}  \{\{v,n_v(i)\}\}$.
\end{lemma}


\vspace{0.35cm}

\boxtext{
\begin{tabbing}
else\= else\= else\= \kill
Algorithm $1$: {\tt Flood (Round Robin)}\\
(Input: max. deg. $\Delta$, own deg. $\Delta_v$, distance $d$,\\
\>\> \  neighbors $n_v(1), \ldots, n_v(\Delta_v)$, rumors $R_v$)\\
\\
REPEAT $d$ times\\
\>$R' = \emptyset$\\
\>FOR $t=1$ to $\Delta$\\
\>\>IF $t \leq \Delta_v$ THEN\\
\>\>\>exchange rumors in $R_v$ with $n_v(t)$\\
\>\>ELSE wait\\
\>\>add all received rumors to $R'$\\
\>$R_v = R_v \cup R'$
\end{tabbing}
}

\vspace{0.35cm}

\begin{proof}[Proof of \Cref{lem:flooding}]
We denote with $R_v(i)$ the set of tokens known to node $v$ at the beginning of iteration $i$. 
The rumors collected in $R'$ during iteration $i$ by node $v$ are exactly the rumors exchanged with nodes neighboring
$v$ in $G'$ since for each undirected link $\{u,v\} \in E'$ either $v$ or $u$ initializes a bidirectional rumor exchange. 
We thus get that for any node $v$ and any iteration $i$ we have $R_v(i+1) = \bigcup_{u \in \Gamma_{G'}(v)} R_u(i)$.
Note furthermore, that for every $v$ and every $k$ we also have $\Gamma_{G'}^{k+1}(v) = \bigcup_{u \in \Gamma_{G'}(v)} \Gamma_{G'}^k(u)$.
Now, using induction on the number of iterations with these two statements we directly obtain that $R_v(d) = \bigcup_{u \in \Gamma_{G'}^d(v)} R_u$ as asserted.
\end{proof}

\subsection{A Simple Randomized {$1$-Local} Broadcast Gossip Algorithm}

In this part we present a simple randomized gossip protocol. The algorithm and even more its analysis are inspired by~\cite{random} but are arguably simpler and more natural. 
 To further simplify the presentation we did not optimize the running time of the algorithm presented here. 

The randomized gossip protocol does the following for each node $v \in V$ in parallel:

\vspace{0.35cm}

\boxtext{
\begin{tabbing}
else\= else\= else\= \kill
Algorithm $2$: {\tt Randomized Gossip}\\
\\
$R_v = v$\\
WHILE $\Gamma_G(v) \setminus R_v \neq \emptyset$\\
\>pick $\Theta(\log^2 n)$ random edges to $\Gamma(v) \setminus R_v$\\
\>$d = \Theta(\log^2 n)$; $E'$ = all (newly) picked edges\\
\>{\tt Flood} rumors in $R_v$ along $E'$-edges for $d$-hops\\
\>add all received rumors to $R_v$

\end{tabbing}
}

\vspace{0.35cm}

It is clear by construction that Algorithm $2$ correctly solves the {$1$-local} broadcast problem since a node continues contacting its neighbors until it has received the rumors from all of them. The next lemma proves that with high probability Algorithm $2$ is furthermore very efficient. 

\begin{lemma}\label{lem:running time randomized gossip}
With high probability Algorithm $2$ takes at most $4 \log n$ iterations and solves the {$1$-local} broadcast problem in $O(\log^6 n)$ rounds (or $O(\log^5 n)$ rounds if only newly picked links are used during the flooding). 
\end{lemma}

To proof \Cref{lem:running time randomized gossip} we need the following two lemmas about expanders (for the definition of the expansion used here we refer to~\cite{random}):

\begin{lemma}[Lemma 3.1 in~\cite{random}]\label{lem:expanderdecomp}
Every graph can be partitioned into disjoint node subsets such that:
\begin{itemize}
	\item Any subset forms a $\Omega(1/\log n)$-expander (when adding edges leaving the subset as self-loops).
	\item At least a third of all edges are intra-partition edges, i.e., both endpoints lie in one subset of the partition. 
\end{itemize}
\end{lemma}

\begin{lemma}\label{lem:diameterofansubsampledexpander}
Let $G$ be an $n$-node graph and expansion $\Phi$ and let $T = \Omega(\frac{\log n}{\Phi})$. Suppose for every vertex $v$ we uniformly sample $T$ neighboring edges (with replacement) and let $G'$ be the subgraph of $G$ consisting of the union of all selected edges. With high probability $G'$ has diameter at most $T$. 
\end{lemma}
\begin{proof}
This follows directly from the result of~\cite{uniformfinal} that uniform gossip solves the global broadcast problem in $G$ in $T$ steps. To see this we note that $G'$ can be seen as the graph of edges initiated during such a run of the uniform gossip protocol. Furthermore, since each message travels at most one step in each round and every node learns about all messages in $T$ steps the graph $G'$ has diameter at most $T$. 
\end{proof}

The result of~\cite{uniformfinal} which we used to prove \Cref{lem:diameterofansubsampledexpander} was also used in~\cite{random}. Interestingly, it seems much stronger than \Cref{lem:diameterofansubsampledexpander} and we suspect that \Cref{lem:diameterofansubsampledexpander} itself can be proved using simpler methods. One way would be to use sparsification results to show that under the specified subsampling any $\Phi$-expander maintains its expansion and thus also its diameter of $O(\frac{\log n}{\Phi})$.

We are now ready to proof \Cref{lem:running time randomized gossip}:

\begin{proof}
Note that since flooding for $d$ hops is symmetric we have at the beginning of any iteration $i$ that $u$ has not heard from $v$ yet if and only if $v$ has not heard from $u$. In this case we say the edge $\{u,v\} \in G$ is active and we denote the graph of active edges at the beginning of iteration $i$ with $K_i$. To prove that $\log n$ iterations are sufficient we will show that with high probability the number of active edges decreases by a factor of $2/3$ in every iteration. For this, we apply \Cref{lem:expanderdecomp} on $K_i$ to get a partitioning of nodes into subsets that induce $\Phi = O(1/\log n)$ expanders. The sampling of $\Theta(\log^2 n)$ uniformly random new neighbors in Algorithm $1$ now directly corresponds to subsampling each of these expanders in the same way as described in \Cref{lem:diameterofansubsampledexpander}. From \Cref{lem:diameterofansubsampledexpander} we thus get that with high probability the distance between any two nodes in the same partition along the newly established links is at most $O(\log n/\Phi)=O(\log^2 n)$. With high probability we thus get that any intra-partition edge becomes inactive after flooding for $\Theta(\log^2 n)$ hops along the (newly) selected links. Since \Cref{lem:expanderdecomp} guarantees that at least a third of the active edges are intra-partition edges we have established that with high probability at most $- \log_{2/3} n^2 < 4 \log n$ many iterations are needed. 
To determine the total running time we note that each node established at most $\Theta(\log^2 n)$ new links in each iteration. The number of total links established by a node is thus at most $O(\log^3 n)$. This means we can run the flooding protocol, i.e., Algorithm $2$, with $d = \Theta(\log^2 n)$ and $\Delta = O(\log^3 n)$ (or $\Delta = O(\log^2 n)$ if we only use newly established links). The running time for one iteration is thus $O(\log^5 n)$ (or $O(\log^4 n)$) rounds according to \Cref{lem:flooding}. Over at most $4 \log n$ iterations this sums up to a total of $O(\log^6 n)$ (or $O(\log^5 n)$) rounds. 
\end{proof}

\section{The New Deterministic {$1$-Local} Broadcast Gossip Protocol}\label{sec:detsimple}

In this section we give the simplest description of our deterministic gossip protocol for the {$1$-local} broadcast problem. It is identical to Algorithm $2$ presented in the last section except for the following three simplifications (marked in bold in Algorithm $3$):
\begin{itemize}
	\item Instead of $\Theta(\log^2 n)$ new edges per round only {\bfseries one} new edge is added per node.
	\item Instead of $\Theta(\log^2 n)$ hops messages are only flooded for ${\mathbf 2}$ {\bfseries log} $\mathbf n$ hops.
	\item Most importantly, the new edge(s) are not anymore required to be chosen uniformly at random but can be chosen in
	any {\bfseries arbitrary (deterministic)} way. 
\end{itemize}	
Surprisingly, we will show next that this severely stripped down algorithm still takes only $\log n$ iterations to solve the {$1$-local} broadcast problem (now deterministically and always instead of with high probability).\\ 
	
As already described, our deterministic gossip protocol does the following for each node $v \in V$ in parallel:

\vspace{0.35cm}

\boxtext{
\begin{tabbing}
else\= else\= else\= \kill
Algorithm 3: {\tt Deterministic Gossip}\\
\\
$R_v = v$\\
WHILE $\Gamma(v) \setminus R_v \neq \emptyset$\\
\>{\bfseries arbitrarily} pick {\bfseries one} new edge to $\Gamma(v) \setminus R_v$\\
\>$d = {\mathbf 2}$ {\bfseries log} $\mathbf n$; $E'$ = all established links\\
\>{\tt Flood} rumors in $R_v$ along $E'$-edges for $d$-hops\\
\>add all received rumors to $R_v$
\end{tabbing}
}

\vspace{0.35cm}

\begin{lemma}\label{lem:running time deterministic gossip}
Algorithm $3$ takes at most $\log n$ iterations and solves the {$1$-local} broadcast problem in at most $2 \log^3 n$ rounds.
\end{lemma}

We first need to define binomial trees for our analysis: 

\begin{definition}
A binomial tree of order $2^i$ or short $i$-tree is a rooted depth $i$ tree on $2^i$ nodes that is inductively defined as follows:
A $0$-tree consists of a single node. For any $i \geq 0$ a $i+1$-tree is formed by 
taking two $i$-trees, connecting their roots and declaring one of the roots as the new root.
\end{definition}

\begin{figure}[ht]
\begin{center}
\includegraphics[width=0.75\textwidth]{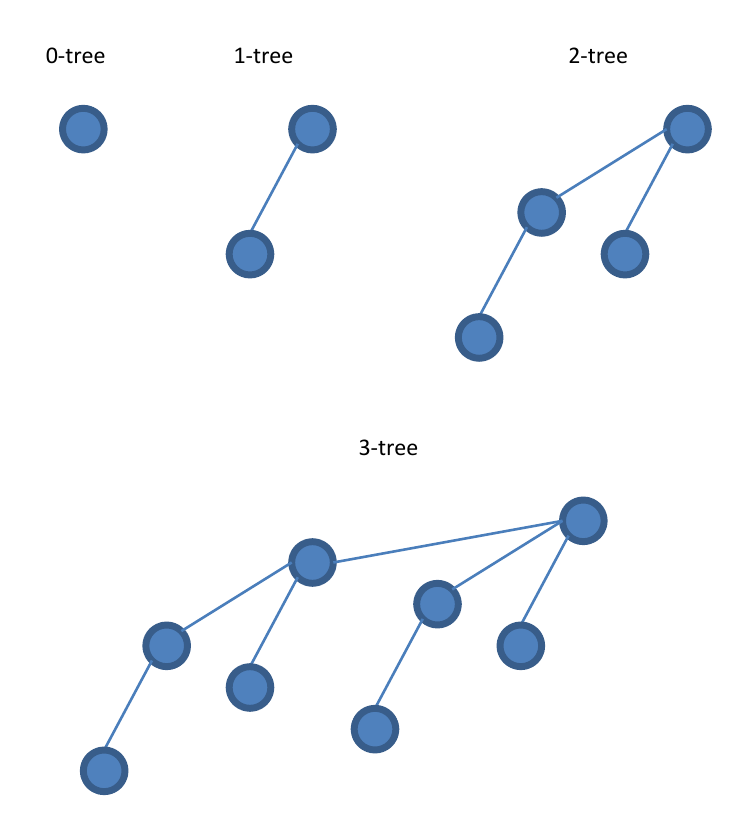}
  \caption{$i$-trees for $i \in \{0,1,2,3\}$. \label{fig:trees}}
\end{center}
\end{figure}

To prove Algorithm $3$ efficient we will use a short and simple inductive argument that for any node that has not terminated at iteration $i$ it is possible to find a $i$-tree rooted at it in $G$. Since $i$-trees grow exponentially in $i$ this limits the number of iterations to $\log n$. 


\begin{lemma}\label{lem:induction statement}
Consider the beginning of any iteration $0 \leq i \leq \log n$ in Algorithm $3$ and $H_i$ be the graph of all edges used until then.
Suppose there is a node $v_0$ with $k$ missing rumors, that is, $\Gamma_G(v_0) \setminus R_{v_0} = \{v_1, \ldots, v_k\}$. Then there are $k+1$ many $i$-trees $\tau_0,\ldots,\tau_{k}$ as subgraphs in $H_i$ rooted at $v_0,v_1,\ldots,v_k$ respectively such that $\tau_0$ is vertex disjoint from $\tau_j$ for any $0 < j \leq k$.
\end{lemma}
\begin{proof} 
We prove the lemma by induction on $t$. The base case for $i=0$ follows directly from the fact that each node forms its own $0$-tree. For the inductive step we assume a vertex $v_0$ which at the beginning of iteration $i+1 \leq \log n$ is still active. Let $u_0$ be the vertex contacted by $v_0$ in iteration $i$. By induction hypotheses in the beginning of iteration $i$ there was an $i$-tree rooted at $v_0$ and a vertex disjoint $i$-tree rooted at $u_0$. These two trees together with the new edge $\{v_0,u_0\}$ form the new $i+1$-tree $\tau_0$ in $H_{i+1}$. 

Next, we note that the symmetry of flooding for $d$ hops ensures that whenever an ID $a$ gets added to $R_b$ the ID $b$ also gets added to $R_a$. Therefore if $\Gamma_G(v_0) \setminus R_{v_0} = \{v_1, \ldots, v_k\}$ at the beginning of iteration $i+1$ it must be that all $v_j \in \Gamma_G(v_0) \setminus R_0$ also have $v_0 \in \Gamma_G(v_j) \setminus R_{v_j}$. Every node $v_j \in \Gamma_G(v_0) \setminus R_0$ was therefore also active at the beginning of round $i$ and must have chosen an edge to a node $u_j$. Similarly as done for $v_0$ we can find an $i+1$-tree $\tau_j$ that consists of the $i$-trees rooted at $v_j$ and $u_j$ at iteration $i$ and the edge $\{v_j,u_j\}$. It only remains to show that $\tau_0$ and $\tau_j$ are node disjoint for all $j$. Assume for sake of a contradiction that there is a $j$ such that $\tau_0$ and $\tau_j$ share a node. In this case there is a path in $H_{i+1}$ from $v_0$ to $v_j$ of length at most $2 \log n$ as the depths of both $\tau_0$ and $\tau_j$ is at most $\log n$. But $H_{i+1}$ is the graph along which IDs and rumors are flooded for $2 \log n$ hops during iteration $i$. Thus $v_j$ would be in $R_{v_0}$ at the beginning of iteration $i+1$ -- this is the desired contradiction that completes the proof. 
\end{proof}

\begin{proof}[Proof of \Cref{lem:running time deterministic gossip}]
Algorithm $3$ correctly solves the {$1$-local} broadcast problem by construction since it keeps contacting new neighbors until it has received the rumors from all of them. \Cref{lem:induction statement} furthermore proves that if the protocol is not done after the iteration $\log n$ at the beginning of the next iteration we can find two neighbors that do not know of each other and two node-disjoint $(\log n)$-trees as subgraphs of $G$. Since the number of nodes is $n$ this is impossible and shows that Algorithm $3$ performs at most $\log n$ iterations. In each iteration at most one new link is established per node for a total of at most $\log n$ links. Flooding for $d = 2 \log n$ hops in the graph of all established links using Algorithm $2$ with $\Delta = \log n$ thus takes $2 \log^2 n$ rounds. Over $\log n$ iteration this accumulates to $2 \log^3 n$ rounds in total. 
\end{proof}

\section{More Efficient Broadcast Protocols}\label{sec:more}

In this section we show how to tweak the deterministic gossip protocol from \Cref{sec:detsimple} to achieve faster {$k$-local} broadcast protocols. In particular, in this section we prove \Cref{thm:main} by showing how to solve the {$k$-local} broadcast problem deterministically in $2(k + \log n)\log n$ rounds. 

\subsection{Faster {$1$-Local} Broadcast via Deterministic Tree Gossip}

First, we speed-up our solution for the {$1$-local} broadcast problem by replacing the flooding step in Algorithm $3$. We use two key observations: 

\begin{itemize}
	\item The flooding steps in Algorithm $3$ are only performed to ensure that in every iteration $i$ any node $v_0$ picks a new neighbor whose $i$-tree does not intersect with its own $i$-tree.
	\item The structure of these $i$-trees allows for spreading rumors within the trees faster than using the round robin flooding procedure. 
\end{itemize}

To better understand the structure of the $i$-trees constructed in the proof of \Cref{lem:induction statement} we will give an alternative construction. For any node $v_0$ that has not terminated until iteration $i$ we construct its $i$-tree $\tau_0$ as follows: 

\smallskip

The root of $\tau_0$ is $v_0$ and its children are the nodes $u_1,\ldots,u_{i-1}$ contacted by $v_0$ in the iterations up to $i$. For each of these child nodes $u_{i'}$ we then attach as children all nodes $w_1,\ldots,w_{i'-1}$ contacted by node $u_{i'}$ in the iteration up to $i'$. We continue inductively for each of these nodes $w_{i''}$ . 

\smallskip

It is easy to see that this produces the same $i$-tree as the one constructed in the proof of \Cref{lem:induction statement}. In addition to being helpful for our proofs this construction has a nice interpretation. The $i$-tree $\tau_0$ in \Cref{lem:induction statement} can be seen as a witness structure that certifies and explains why node $v_0$ was active until iteration $i$, namely:

\smallskip

Node $v_0$ did not terminate until iteration $i$ because there were the neighbors $u_1,\ldots,u_{i-1}$ that were active and unknown to $v_0$ at time $1$ to $i-1$ respectively resulting in $v_0$ contacting them. Each of these nodes $u_{i'}$ on the other hand was still active and did not contact $v_0$ itself until iteration $i'$ because of its neighbors $w_1,\ldots,w_{i'-1}$ that were active and unknown to $u_{i'}$ at time $1$ to $i'-1$ respectively resulting in $u_{i'}$ contacting them, and so on. 

\smallskip

With this interpretation the proof of \Cref{lem:running time deterministic gossip} essentially says that Algorithm $3$ cannot have an node $v_0$ that is still active after $\log n$ iterations since any explanation $\tau_0$ for why it is still active would have to blame more nodes than exist in the network.  

\smallskip

Next we will show how to exploit the structure of these $i$-trees. The deterministic gossip protocol that does this performs the following for each node $v \in V$ in parallel:

\vspace{0.35cm}

\boxtext{
\begin{tabbing}
else\= else\= else\= \kill
Algorithm 4: {\tt Deterministic Tree Gossip}\\
\\[-0.1cm]
$R = v$\\
FOR $i=1$ UNTIL $\Gamma(v) \setminus R = \emptyset$\\
\>link to any new neighbor $u_t \in \Gamma(v) \setminus N$\\
\>$R' = v$\\
\>PUSH: For $j=i$ downto $1$:\\
\>\>exchange rumors in $R'$ with $u_{j}$\\
\>\>add all received rumors to $R'$\\
\>PULL: For $j=1$ to $i$:\\
\>\>exchange rumors in $R'$ with $u_{j}$\\
\>\>add all received rumors to $R'$\\
\>$R'' = v$\\
\>perform PULL, PUSH with $R''$\\
\>$R = R' \cup R''$
\end{tabbing}
}

\vspace{0.35cm}

{\noindent \bfseries Remark:\\}
Note that Algorithm $4$ does not require knowledge of the network size $n$. This can also be achieved in Algorithm $3$ if one instead of flooding for $2 \log n$ hops in each iteration only floods for $2i$ hops in the $i$th iteration. This also speeds up the running time by a factor of two and avoids the guess-and-double strategy remarked in \Cref{sec:model} for the case that $n$ is unknown.

\begin{theorem}\label{thm:fasteralgorithm}
Algorithm $4$ solves the {$1$-local} broadcast problem in $\log n$ iterations and less than $2 (\log n +1)^2$ rounds. 
\end{theorem}

The proof is essentially the same as for \Cref{lem:running time deterministic gossip} except for the use of the following lemma:

\begin{lemma}\label{lem:treepipelining}
Suppose $u$ and $v$ are two nodes that are active at the beginning of iteration $i$ and that $\tau_u$ and $\tau_v$ are their ordered $i$-trees respectively. Then, after the PUSH exchanges in iteration $i$ all nodes in $\tau_u$ have learned about $u$ (and all nodes in $\tau_v$ have learned about $v$). Furthermore, if $\tau_u$ and $\tau_v$ are not node disjoint then after the PUSH-PULL exchanges in iteration $i$ the node $u$ has learned about $v$ (and vice versa). 
\end{lemma}
\begin{proof}
We note that by construction all paths from the root $u$ to a any other node in $\tau_u$ follow tree edges in decreasing order. During the first PUSH sequence edges are activated in decreasing order which thus pipelines the rumor of $u$ from the root to all nodes in $\tau_u$. Similarly, rumors known to any nodes in $\tau_u$ gets pipelined towards the root $u$ during the first PULL sequence. The same is true by symmetry for $v$ and $\tau_v$. Now, if $\tau_u$ and $\tau_v$ share a node $y$ then $y$ will learn about $u$ and $v$ during the first PUSH sequence of iteration $i$ and then forward this information to $u$ and $v$ during the first PULL sequence -- informing both nodes about each other. 
\end{proof}

\begin{figure}[ht]
\begin{center}
\includegraphics[width=0.8\textwidth]{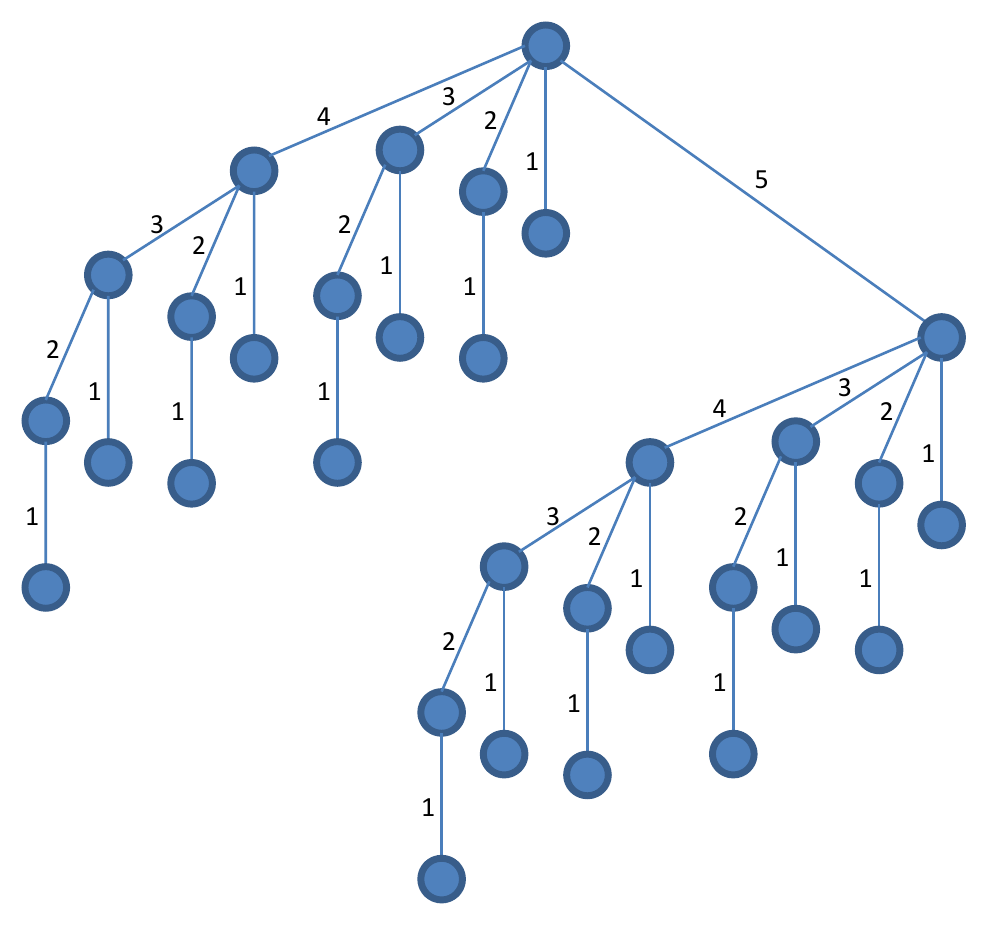}
 \caption{A $5$-tree. The labels on the edges denote at what time they were added. These edges are activated from low to high during a PUSH exchange and activated from high to low in a PULL exchange. Note that the path from the root to any node follows edges in decreasing order. \label{fig:5tree}}
\end{center}
\end{figure}

\begin{proof}[\Cref{thm:fasteralgorithm}]
We use the exact statement as in \Cref{lem:induction statement} to show that only $\log n$ iterations are performed by Algorithm $4$. The first thing to check is that Algorithm $4$ still maintains symmetry with regards to which nodes knows which. This is achieved by the the fact that the PULL-PUSH sequence of the second part is exactly the reversal of the PUSH-PULL sequence of the first part of any iteration. As such, if a node $v$ learns about a node $u$ in the PUSH-PULL sequence of exchanges then $u$ will learn about $v$ during the PULL-PUSH exchanges and vice versa. By taking $R$ to be the union of $R'$ and $R''$ it is clear that we indeed get $u \in R_v$ if and only if $v \in R_u$. The rest of the proof of \Cref{lem:induction statement} goes exactly the same except that we relay on \Cref{lem:treepipelining} for the fact that if $v_0$ and $v_i$ do not know about each other then their $i$-trees must be disjoint. This completes the proof that at most $\log n$ iterations are performed. To determine the total number of rounds needed we note that in the $i$th iteration exactly $4i$ exchanges are performed. This leads to a total round complexity of $\sum_{i=1}^{\log n} 4 i = 2 \log n (\log n + 1)$.
\end{proof}

\subsection{Speeding up the {$k$-Local} Broadcast for $k>1$}

What remains to show for \Cref{thm:main} is how to use our algorithms for the $k$-local broadcast problem for $k>1$. As remarked in \Cref{sec:problemdef} a straight forward way to use Algorithm $2$ -- $4$ for the $k$-local broadcast problem is simply applying them $k$ times. Using our Algorithm $4$ this leads to a $4 k \log^2 n$ solution. We can improve upon this by realizing that in all algorithms the {$1$-local} broadcast problem is actually solved completely during the flooding or tree-broadcast in the last iteration. All prior iterations are only needed to guarantee that the links chosen are going to new neighbors. To exploit this, we only use the full blown algorithms for the first {$1$-local} broadcast and then reuse the established links by simply repeating the flooding or tree-broadcast of the last iteration for the remaining local broadcasts. When using Algorithm $4$ for this we need $2 \log n (\log n + 1)$ round for the first local broadcast and only $2 \log n$ rounds for each of the $k-1$ remaining ones. This leads to the running time claimed in \Cref{thm:main}.

\section{Natural Gossip Processes and Robustness}\label{sec:Natural Gossip Processes and Robustness}

In addition to their message and time efficiency gossip algorithm have also been studied because of their 
naturalness and robustness.


Unfortunately, while the algorithms in \cite{weakcond,random} deal much better with bottlenecks in the topology
than the uniform gossip protocol they are also much less natural and robust. This holds In particular, for the local
broadcast algorithm of \cite{random} which crucially relies on a very unnatural reversal step to guarantee correctness.
This algorithm can furthermore fail completely if only one link is temporarily down in one round or if due to a slight
asynchrony the order of two exchanges gets switched. 

In this section we show that the algorithms and analysis presented in this paper are robust and can
be seen and phrased as a natural process. The latter is best demonstrated
by considering as an example the following \emph{surprisingly accurate/realistic} social setting:

\medskip

\emph{Interpret nodes as curious but somewhat shy persons that want to know (everything) about their neighbors but only rarely have the courage to approach a person they know nothing about; instead they prefer to talk to neighbors they have already approached before and exchange rumors/information about others.} 

\medskip

We will show that our analysis is flexible enough to not just show fast rumor propagation for our specifically designed algorithms but broadly covers a wide variety of high level processes including the one above. In particular, for this setting our analysis implies the following: Rumors spread rapidly in such a social setting as long as a person talks (by a factor of $\log^2 n$) more frequently to a person approached before compared to approaching a new neighbor (whose rumor is not known). Interestingly, our result is flexible enough to allow the (social) process according to which nodes choose which new person they find most approachable at any time to be arbitrarily dependent and complex. 

\smallskip

To show the flexibility of our analysis we consider the template given as Algoithm 5. It is coined in terms of iterations each consisting of a linking step and a propagation step. 

\vspace{0.35cm}

\boxtext{
\begin{tabbing}
else\= else\= else\= \kill
Algorithm 5: {\tt Gossip Template}\\
\\
REPEAT\\
\>link to any neighbor whose rumor is not known\\
\>propagate rumors among established links\\
UNTIL all rumors are known
\end{tabbing}
}

\vspace{0.35cm}

To analyze this template we introduce some notation: Let $R_v(t) \subseteq \Gamma(v)$ be the neighbors of node $v$ it knows the IDs and rumors of at the beginning of iteration $t$. Furthermore let $G_t$ be the undirected graph that consists of all edges added in the linking procedure. With this we first show a strong but admittedly relatively technically phrased lemma. 

\begin{lemma}\label{lem:naturalmain}
Let $T_{\min}$ be the maximum number of propagation steps it takes for a nodes $u$ to learn about another node $v$ for which there exists a path of established links of length at most $2 \log n$. Furthermore, for any two nodes $u,v$ let $T_{diff}$ be the maximum number of propagation steps it takes from the time that $u$ knows $v$ until $v$ also knows $u$. With these two parameters Algorithm $5$ takes at most $T \log n$ iterations, where $T = \max \{T_{diff},T_{\min}\}$.
\end{lemma}

\noindent {\bfseries Remark:}

\smallskip

The algorithm in \cite{random} and all algorithms presented in this paper so far keep perfect symmetry of knowledge, that is, if a node $u$ knows the rumor of  $v$ then $v$ also knows the rumor of $u$. This is achieved by flooding for exactly $d$ hops or by carefully reversing the sequence in which edges where chosen. This symmetry is used in all proofs so far. If furthermore turns out to be crucial for the efficiency of the algorithm of \cite{random}: Indeed, \cite{random} gives an example in which introducing a slight asymmetry increases the running time from polylogarithmic to linear. The expander decomposition proof of \cite{random} seems furthermore unsuitable for extensions to a more asymmetric setting. While our algorithms so far also featured perfect symmetry \Cref{lem:naturalmain} shows that both algorithms and our analysis are robust enough to relax this requirement significantly. 

\medskip

While \Cref{lem:naturalmain} shows robustness it does not read too natural or close to the informal description given before. This is remedied by the next two corollaries which give some examples on how it can be used.

\begin{corollary}\label{cor:natural1}
Suppose a rumor dies out / looses credibility (i.e., is not forwarded anymore) after it has been passed around for more than $\lambda$ hops where $\lambda > \log n$. Suppose also that nodes establish links at least every $\alpha$ steps and talk to each established link at least every $\beta$ steps. Then the {$1$-local} broadcast completes after at most $\alpha \beta \lambda$ steps. 
\end{corollary}

Or on an even more concrete example:

\begin{corollary}\label{cor:natural2}
Suppose a rumor dies out after $\Theta(\log n)$ hops. Suppose also that each node independently chooses with probability $p = 1 / \log^2 n$ to talk to a neighbor it has not heard from yet and otherwise talks to a random neighbor it has already contacted before. In this setting the  {$1$-local} broadcast stops after $T = O(\log^4 n)$ iterations with high probability. 
\end{corollary}

Given the above demonstrated flexibility it comes as no surprise that our algorithms are also naturally robust against various kinds of failures. In particular, it is easy to give robust deterministic algorithms for the $k$-local broadcast based on \Cref{lem:naturalmain}. For example, a simple round robin flooding procedure (with or without distance labels to prevent rumors from spreading too far) is naturally robust against any random edge failure rate $\gamma$ with only the necessary $1/(1-\gamma)$-slowdown. Even adversarial permanent failures merely slow down the algorithm slightly as they can not cause more harm then preventing progress made in the iteration the failed edge occurred. We remark that the subsequent but independent work in \cite{robustrandom} also gives a alternative to \cite{random} which is robust against random temporary and random permanent node- and edge-failures. The algorithms in \cite{robustrandom} also solve the {$1$-local} broadcast problem considered in this paper but are still randomized and have a running time of $O(\log^9 n)$ rounds.

\subsection{Proofs}

\begin{proof}[Proof of \Cref{lem:naturalmain}]
We will only look at links established at iteration $i' \equiv 0 \mod T$. We define $H_i$ to be subgraph of $G$ consisting of all edges established in these iterations until iteration $iT$. We then prove \Cref{lem:induction statement} for Algorithm $5$ exactly as before. That is:\\

Consider the beginning of any iteration $iT$ for $i \leq \log n$ in Algorithm $5$ and any active node $v_0$.
If $\Gamma(v) \setminus R = \{v_1, \ldots, v_k\}$ then there are $k+1$ many $t$-trees $\tau_0,\ldots,\tau_{k}$ as subgraphs in $H_{i+1}$ rooted at $v_0,v_1,\ldots,v_k$ respectively such that $\tau_0$ is vertex disjoint from $\tau_i$ for any $0 < i \leq k$.\\

We prove the lemma by induction on $i$. The base case for $i=0$ follows directly from the fact that each node forms its own $0$-tree. For the inductive step we assume a vertex $v_0$ which at the beginning of iteration $(i+1)T$ for $i+1 \leq \log n$ is still active. Let $u_0$ be the vertex contacted by $v_0$ in iteration $iT$. By induction hypotheses in the beginning of iteration $iT$ there was a $i$-tree rooted at $v_0$ and a vertex disjoint $i$-tree rooted at $u_0$. These two trees together with the new edge $\{u,v\}$ form the new $i+1$-tree $\tau_0$. 

All neighbors $v_i$ of $v_0$ that are in $\Gamma(v) \setminus R$ at beginning of iteration $(i+1)T$ can not have know about $v_0$ at the beginning of iteration $iT$ since otherwise $v_0$ would know about $v_i$ at iteration $iT + T_{diff} \leq (i+1)T$. They were therefore also active at the beginning of iteration $iT$ and must have chosen an edge to a node $u_i$. Similarly as done for $v_0$ we can find an $i+1$-tree $\tau_i$ that consists of the $i$-trees rooted at $v_i$ and $u_i$ at iteration $i$ and the edge $\{v_i,u_i\}$. It is clear that $\tau_0$ and $\tau_i$ are node disjoint since otherwise there is a path from $v_i$ to $v_0$ of length at most sum of the depths of $\tau_0$ and $\tau_i$ in $H_{i+1}$. This is at most $2 \log n$ which implies that in the beginning of iteration $iT + T_{\min} \leq (i+1)T$ the node $v_0$ would know about $v_i$ -- a contradiction. 
\end{proof}

\begin{proof}[Proof of \Cref{cor:natural1}]
In this case we get that one iteration corresponds to $\alpha$ steps. We have $T_{diff} = \lambda \beta$ since if $u$ is informed about $v$ then there is a path of length at most $\lambda$ and it takes at most $\beta$ rounds per step until $v$ also knows about $u$. 
\end{proof}

\begin{proof}[Proof of \Cref{cor:natural2}]
First, we note that with high probability every node contacts a new neighbor within any $\alpha = O(\log^3 n)$ rounds. Furthermore, with high probability the number of established neighbors during $O(\log^4 n)$ rounds is at most $O(\log^2 n)$ for every node. Thus, if a path of established links of lengths at most $\Theta(\log n)$ occurs between two node $u$ and $v$ then after $O(\log^3 n)$ steps both nodes will have learned from each other with high probability. According to \Cref{lem:naturalmain} it takes thus at most $\max\{O(\log^3 n), O(\log^3 n)\}\log n = O(\log^4 n)$ steps to complete the {$1$-local} broadcast. 
\end{proof}

\section{Conclusion}

In this paper we presented the first efficient deterministic gossip algorithm for the rumor spreading problem and the $k$-local broadcast problem. In addition to showing that all random choices of a certain gossip algorithm can be replaced by arbitrary deterministic choices our algorithms are also much simpler, more robust, more natural and with a running time of $2(k \log n + \log^2 n)$ faster than previous randomized algorithms. 

One interesting question that remains is whether the running time of $O(\log^2 n)$ for the $1$-local broadcast problem or the $\log n$-local 
broadcast problem can be improved. While we believe that our running time is optimal, at least for deterministic algorithms, we could not prove
such a lower bound. One direction for proving lower bounds even for randomized algorithms would be to use the connections to spanners
given in~\cite{random}. In particular, any gossip algorithm solving, e.g., the $\log n$-local broadcast problem in $T = \log^{1+\delta} n$ rounds
implies the existence of an $(\alpha,\beta)$-spanner with $\alpha = \Theta(\log^{\delta} n)$ and $\beta = \log^{1+\delta} n$ with density $n \log^{1+\delta} n$. No spanner of such quality is known to exist for $\delta < 1$~\cite{Pettie2009}.

The connections to spanners from \cite{random} can also be used to in the opposite direction to see our gossip algorithms for the $1$-local broadcast problem as an extremely efficient distributed constructions of a sparse graph spanner. For example the subgraph of edges used by Algorithm $3$ or $4$ both form a spanner with stretch $2 \log n$ and $n \log n$ edges. This is almost optimal as $\Theta(\log n / \log \log n)$ is the best stretch achievable with this density. While constructions of optimal spanners are known, the simple distributed construction of a good quality spanner in the extremely weak gossip model should be of interest. In particular, the total number of messages used in Algorithm $4$ is with $n \log^2 n$ drastically less than any distributed algorithm in the literature since these LOCAL-algorithms send a messages over each of the potentially $\Theta(n^2)$ many edges in $G$ in every round. Similarly, the total amount of information exchanged by our algorithm is quite low given that there are no distributed algorithms that use less than polynomial size messages\cite{pettie2010distributed}. Even the $4\log^2 n$ round construction time is quite fast given the restricted model. We remark that our construction can also be seen as a deterministic $2\log n$ round algorithm in the LOCAL model since nodes only communicate with other nodes in their $2\log n$ neighborhood. All this, together with the simplicity of our algorithm and its analysis makes our approach an interesting starting point for designing new spanner constructions. One question in this direction would for example be: If one floods only for $k$ instead of $2 \log n$ hops during Algorithm $3$ then a $k$-stretch spanner is computed. What is the density of this spanner when $k < \log n$?

Another question is whether the message size requirement can be reduced. While, as remarked in \Cref{sec:modelandproblem}, this is not reasonable to ask for in the local broadcast problem itself the trivial information argument does not apply to the bit complexity for discovering a sparse spanner via a gossip algorithm. Nevertheless, reducing the message size seems like a hard and quite possibly impossible task. As just mentioned even for the intensely studied case of constructing spanners via randomized algorithms in the less restrictive LOCAL model no algorithms using a subpolynomial message size are known~\cite{pettie2010distributed}.



\bibliographystyle{abbrv}
\bibliography{Gossip}

\end{document}